\documentclass[preprint,12pt]{elsarticle}

\makeatletter
\def\ps@pprintTitle{%
 \let\@oddhead\@empty
 \let\@evenhead\@empty
 \def\@oddfoot{\centerline{\thepage}}%
 \let\@evenfoot\@oddfoot}
\makeatother
\makeatletter
\def\paragraph{\@startsection{paragraph}{4}%
  \z@\z@{-\fontdimen2\font}%
  {\normalfont\bfseries}}
\makeatother

\usepackage[T1]{fontenc} 
\usepackage[english]{babel}
\usepackage{verbatim}
\usepackage{amsmath,amsthm,amssymb}
\usepackage{amsthm}
\usepackage{xcolor}
\usepackage{bm}
\theoremstyle{plain}
\usepackage{graphicx, caption, subcaption}

\usepackage{float}
\usepackage{amsmath, amssymb}
\usepackage{enumitem, latexsym}  
\usepackage{geometry, natbib} 
\usepackage{tabularx}

\usepackage{colortbl}
\usepackage{color}
\definecolor{Gray}{gray}{0.9}
\usepackage{longtable}
\usepackage[acronym]{glossaries}

\usepackage[referable]{threeparttablex}
\renewlist{tablenotes}{enumerate}{1}
\makeatletter
\setlist[tablenotes]{label=\tnote{\alph*},ref=\alph*,itemsep=\z@,topsep=\z@skip,partopsep=\z@skip,parsep=\z@,itemindent=\z@,labelindent=\tabcolsep,labelsep=.2em,leftmargin=*,align=left,before={\footnotesize}}
\makeatother

\usepackage{hhline} 

\usepackage{multirow}
\usepackage{adjustbox}
\usepackage{fancyhdr}
 
\fancypagestyle{plain}{%
   \fancyhf{}
   \fancyfoot[C]{\iffloatpage{}{\thepage}}
   }
\newtheorem{thm}{Theorem}[section]

\newtheorem{prop}[thm]{Proposition}
\theoremstyle{definition}

\theoremstyle{remark}
\newtheorem{remark}[thm]{Remark}

\usepackage{upref}
\usepackage{graphicx, caption, subcaption}
\usepackage{color}
\usepackage{listings}
\usepackage{epstopdf}

\usepackage{mathrsfs}
\usepackage{multirow}
\providecommand{\mathscr}[1]{\mathcal{#1}}

\usepackage{graphicx}
\usepackage{booktabs}
\usepackage{hyperref}
\usepackage{geometry, natbib} 
\usepackage{adjustbox}
\usepackage[autostyle]{csquotes} 
\usepackage[normalem]{ulem}

\def\R{{\Bbb R}}

\def\P{{\Bbb P}}

\def\NN{{\Bbb N}}

\def\1{{\bf 1}}

\def\A{{\cal A}}
\def\bP{{\sf P}}
\def\E{{\mathbb{E}}}

\usepackage[toc,page]{appendix}

\theoremstyle{plain}
\newtheorem{Theorem}{Theorem}

\newtheorem{Lemma*}{Lemma}

\theoremstyle{definition}
\newtheorem{Definition}{Definition}
\theoremstyle{remark}

\def\correspondingauthor{\footnote{Corresponding author.}}
\newcommand{\UniBA}{%
   Universit\`{a} degli Studi di Bari ``Aldo Moro'',
   Department of~Economics and Finance,
   Largo Abbazia S.~Scolastica,
   Bari, I-70124 Italy}

\begin{document}

\title{Modeling Volatility of Disaster-Affected Populations: A Non-Homogeneous Geometric-Skew Brownian Motion Approach}

\author[ssm]{Giacomo Ascione} 
\address[ssm]{Scuola Superiore Meridionale, Largo S. Marcellino 10, Napoli, 80138, \texttt{g.ascione@ssmeridionale.it}}

\author[uniro]{Michele Bufalo} 
\address[uniro]{Universit\`a degli Studi di Roma "La Sapienza" - Department of Methods and Models for Economics, Territory and Finance, Via del Castro Laurenziano 9, Roma, I-00185, Tel. +39-06-49766903, \texttt{michele.bufalo@uniroma1.it}}

\author[uniba]{Giuseppe Orlando\correspondingauthor{}}
\address[uniba]{\UniBA, Tel. +39 080 5049218, \texttt{giuseppe.orlando@uniba.it}}

\begin{abstract}
This paper delves into the impact of natural disasters on affected populations and underscores the imperative of reducing disaster-related fatalities through proactive strategies. On average, approximately 45,000 individuals succumb annually to natural disasters amid a surge in economic losses. The paper explores catastrophe models for loss projection, emphasizes the necessity of evaluating volatility in disaster risk, and introduces an innovative model that integrates historical data, addresses data skewness, and accommodates temporal dependencies to forecast shifts in mortality. To this end, we introduce a time-varying skew Brownian motion model, for which we provide proof of the solution's existence and uniqueness. In this model, parameters change over time, and past occurrences are integrated via volatility.
\end{abstract}

\begin{keyword} 
Skew Brownian motion; Natural disasters; Time dependency
JEL Classification: C22; Q54; C53

\MSC[2020] 60J65; 91B84; 91B76; 92Fxx 

\end{keyword}

\maketitle

\newpage

\section{Introduction}

This paper deals with total people affected by natural disasters because, as mentioned by \cite{Ritchie2022}, while preventing infrequent, high-impact events might be difficult, reducing overall disaster-related deaths is achievable through early prediction, resilient infrastructure, emergency preparedness, and response systems. Particularly vulnerable are individuals with low incomes; improving living standards and response mechanisms in these areas will be crucial in preventing natural disaster-related deaths in the future.
The death toll from natural disasters can vary significantly from year to year, with some years having few deaths while others experience major disasters causing numerous fatalities. On average over the past decade, around 45,000 people worldwide die annually due to natural disasters, making up about 0.1\% of global deaths. In some years, deaths can be quite low, even as low as 0.01\% of total fatalities. However, major events like famines, earthquakes, and tsunamis have caused the death toll to exceed 200,000 in certain years, constituting over 0.4\% of deaths \cite{Ritchie2022}.  

Adjusted for inflation, economic losses from natural disasters have increased over recent decades, with the number of significant loss-causing events tripling since the 1980s. Notable examples include the Northridge earthquake (1994), Kobe earthquake (1995), Asian tsunami (2004), Hurricane Katrina (2005), Japan earthquake and tsunami (2011), and Hurricane Harvey (2017) \cite{Botzen2019}. 
On the contrary, while in the early to mid-20th century annual deaths from disasters often exceeded a million, in recent decades have seen a substantial decline, with most years witnessing fewer than 20,000 deaths and even fewer in the last decade, even during high-impact events. This decline is remarkable, considering population growth  \cite{Ritchie2022}, and it is explained by the fact that economic development functions as implicit insurance against shocks from natural events \cite{Kahn2005}.

Catastrophe models employ geographic information systems (GIS) to predict potential losses caused by natural disasters by simulating hypothetical hazard characteristics at specific locations \cite{deMoel2015}. For instance, flood hazard maps depict flood-prone areas, inundation depths, and flow velocities. These characteristics help calculate damage to exposed property based on vulnerability assumptions. These models estimate damage across various intensities and probabilities, yielding annual expected damage. While focused on property damage estimation, catastrophe models also project impacts on populations and potential casualties resulting from specific natural disasters \cite{Jonkman2008}. This is because  population and economic growth remain key drivers of these losses \cite{Botzen2019}.

With that said, assessing volatility is crucial from both a risk management and mitigation perspective.
Embracing a risk-averse standpoint and incorporating the volatility linked to disaster risk carries significant implications for evaluating core projects and risk management strategies. Considering natural disaster volatility is vital for evaluating both secondary risk management projects and primary investments \cite{Kreimer2010}. Nations facing high natural hazard exposure and limited coping capacities should carefully weigh disaster risk and volatility in project decisions. Accounting for outcome volatility increases the appeal of risk transfer measures, especially for risk-averse situations in fact, while cost-efficiency through Cost-Benefit Analysis (CBA) is crucial, it shouldn't be the sole criterion. CBA aids efficient fund allocation for more resilient development and should be integral to decision-making for a prevention-focused culture, potentially yielding substantial intangible benefits \cite{Kreimer2010}.

After discussing why we focus on the affected people and their volatility, the next step involves the modeling.  Various stochastic models have been developed to understand sudden shifts in mortality rates (often to price catastrophic bonds). These models differ in how they represent mortality jumps' characteristics. For instance, Cox et al. \cite{Cox2006} combined geometric Brownian motion and compound Poisson processes to model age-adjusted rates, while Chen and Cox \cite{Chen2009} used a normal distribution for jump severity. Chen and Cummins \cite{Chen2010} integrated two types of jumps, and Deng et al. \cite{Deng2012} explored a double-exponential jump process. Liu and Li \cite{Liu2015} focused on the age-related impact of mortality jumps.

These models assume either annual mortality jumps or utilize a Poisson process for jump frequencies. However, due to the rarity and significance of such events, predicting the timing and frequency of future catastrophic events, and hence mortality jumps, remains uncertain \cite{Chen2009}. Historical data, though, can offer insights despite the Poisson process's limitation stemming from the memoryless nature of the exponential distribution. For the reasons mentioned above, we introduce a time-inhomogeneous skew Brownian motion model in which parameters vary with time, and past occurrences are incorporated through volatility. This design allows the model to account for data skewness. Furthermore, we provide the proof of the existence and uniqueness of the solution. \\


This article is divided in the following parts. Section \ref{Sec:Data} describes the data, offers visualizations, and presents the statistical characteristics. Section \ref{Sec:Model} introduces the time-inhomogeneous geometric skew Brownian motion model in terms of a stochastic differential equation and offers proof of the existence and uniqueness of the solution. Section \ref{Sec:Results} displays the empirical results, while the final Section \ref{Sec:Conclusion} provides the conclusion.

\section{Data} \label{Sec:Data}

The Emergency Events Database (EM-DAT), maintained by the Centre for Research on the Epidemiology of Disasters (CRED), is the primary data source for natural disasters \cite{Cred2023}. Alternative databases like NatCatSERVICE and Sigma from reinsurance companies Munich Re and Swiss Re are less commonly used due to limited accessibility. The intensity measures from EM-DAT tend to correlate with GDP per capita, a central variable in research, as losses are more substantial and accurately documented in developed countries \cite{Botzen2019}.
"Our World in Data" \cite{Ritchie2022} offers visualizations and condensed content derived from EM-DAT data.
Figure \ref{Fig:AnnualDeaths} displays the decadal average of the annual number of deaths resulting from various disasters highlighting the downtrend. Figure \ref{Fig:NumberDeathsMap} depicts a map illustrating the number of deaths from natural disasters in 2022, emphasizing that the Indian subcontinent, Pakistan, and Afghanistan are the most high-risk regions. Additionally, Figure \ref{Fig:NatCatClass} portrays the total number of people affected by natural disasters from 1900 to 2023, revealing a pattern of mean reversion. Lastly, Table \ref{T:StatChar} highlights the high volatility, tailedness, and leptokurtic nature of the analyzed time series.

\begin{figure}[!ht]
	\centering
	\includegraphics[width=.65\textwidth]{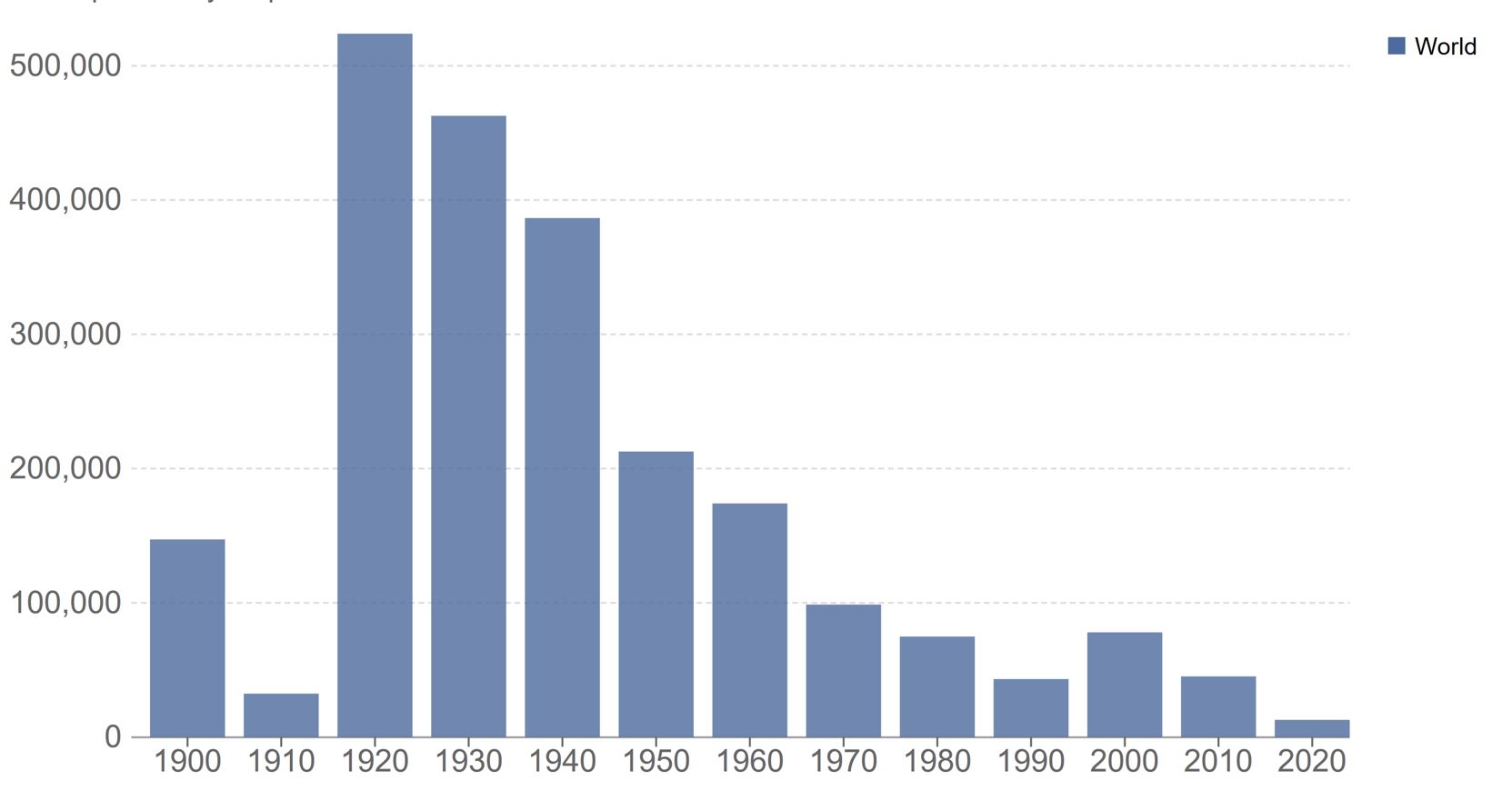}
	\caption{Geophysical, meteorological and climate events including earthquakes, volcanic activity, landslides, drought, wildfires, storms, and flooding. Decadal figures are measured as the annual average over the subsequent ten-year period i.e. figures for '1900' indicate the average from 1900 to 1909, and '1910' represents the average from 1910 to 1919, and so on. Source \cite{Ritchie2022, Cred2023}}
	\label{Fig:AnnualDeaths}
\end{figure}

\begin{figure}[!ht]
	\centering
	\includegraphics[width=.65\textwidth]{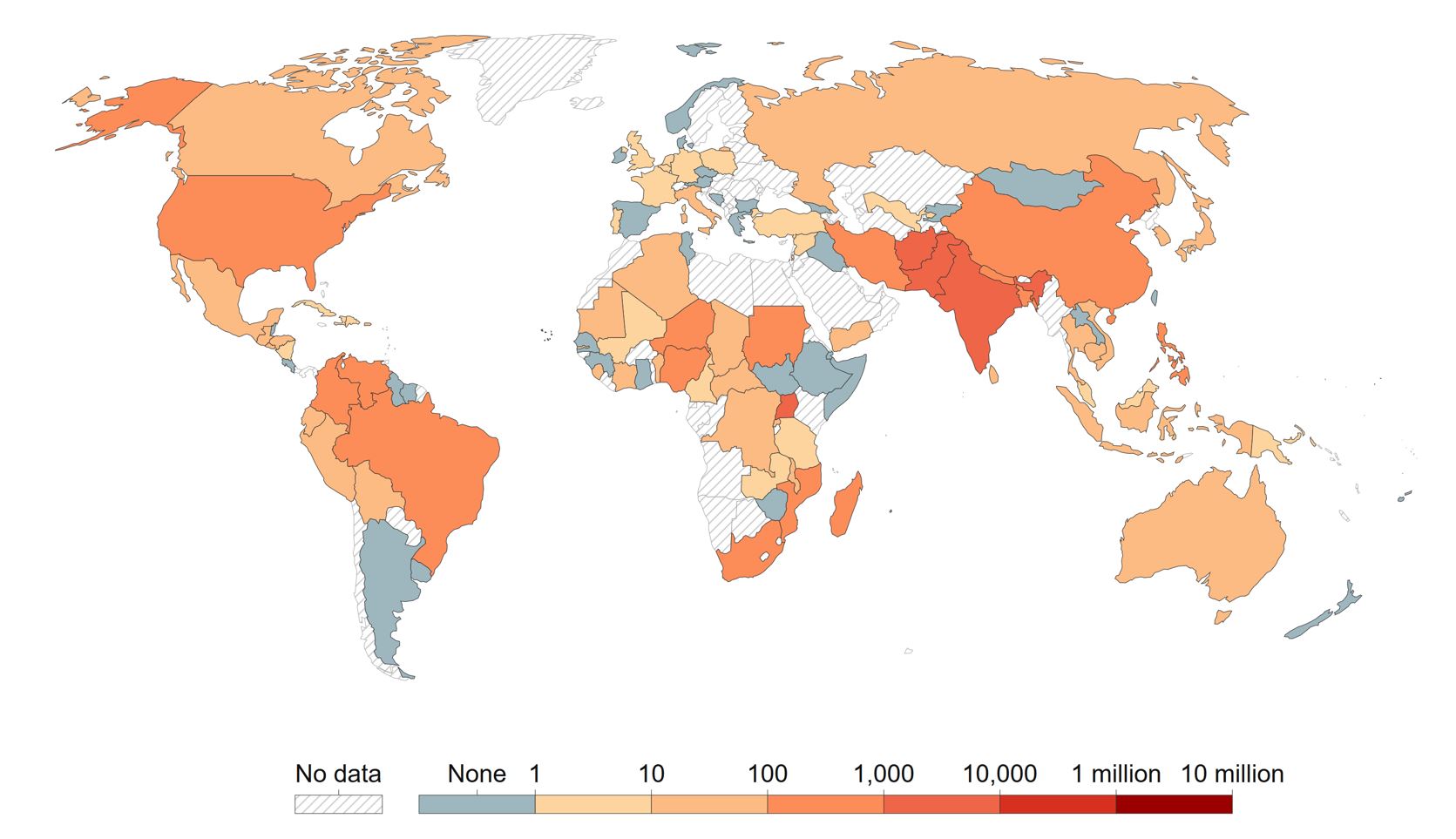}
	\caption{Number of deaths from natural disasters in 2022. Disasters include all geophysical, meteorological and climate events including earthquakes,volcanic activity, landslides, drought, wildfires, storms, and flooding. Source \cite{Ritchie2022, Cred2023}}
	\label{Fig:NumberDeathsMap}
\end{figure}

\begin{figure}[!ht]
	\centering
	\includegraphics[width=.95\textwidth]{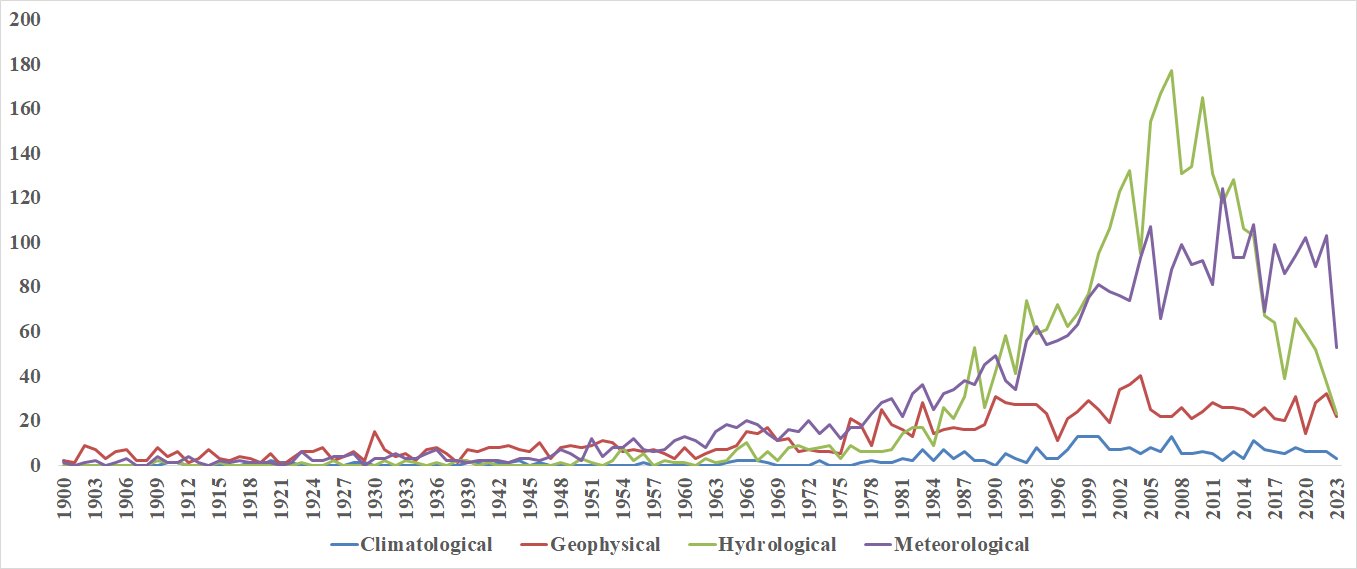}
	\caption{The total number of people affected includes deaths, injuries, those requiring assistance, and individuals displaced from their homes due to natural disasters over the period spanning from 1900 to 2023. Source \cite{Cred2023}}
	\label{Fig:NatCatClass}
\end{figure}

\begin{table}[!ht]
\centering
\adjustbox{max width=.5\textwidth}{%
\begin{tabular}{@{}clll@{}}
\toprule
\multicolumn{4}{c}{Total People Affected Time Series}     \\ \midrule
Mean        & \multicolumn{1}{c}{Std. Dev.} & \multicolumn{1}{c}{Skewness} & \multicolumn{1}{c}{Kurtosis} \\
\multicolumn{1}{l}{223,664.919} & 561,471.353    & 4.102 & 18.495    \\ \bottomrule
\end{tabular}}
\caption{Total number of people affected includes deaths, injuries, those requiring assistance, and individuals displaced from their homes due to natural disasters over the period spanning from 1900 to 2023. Source \cite{Cred2023} } \label{T:StatChar} 
\end{table}

\begin{figure}[!htbp]
	\centering
\includegraphics[width=1\textwidth]{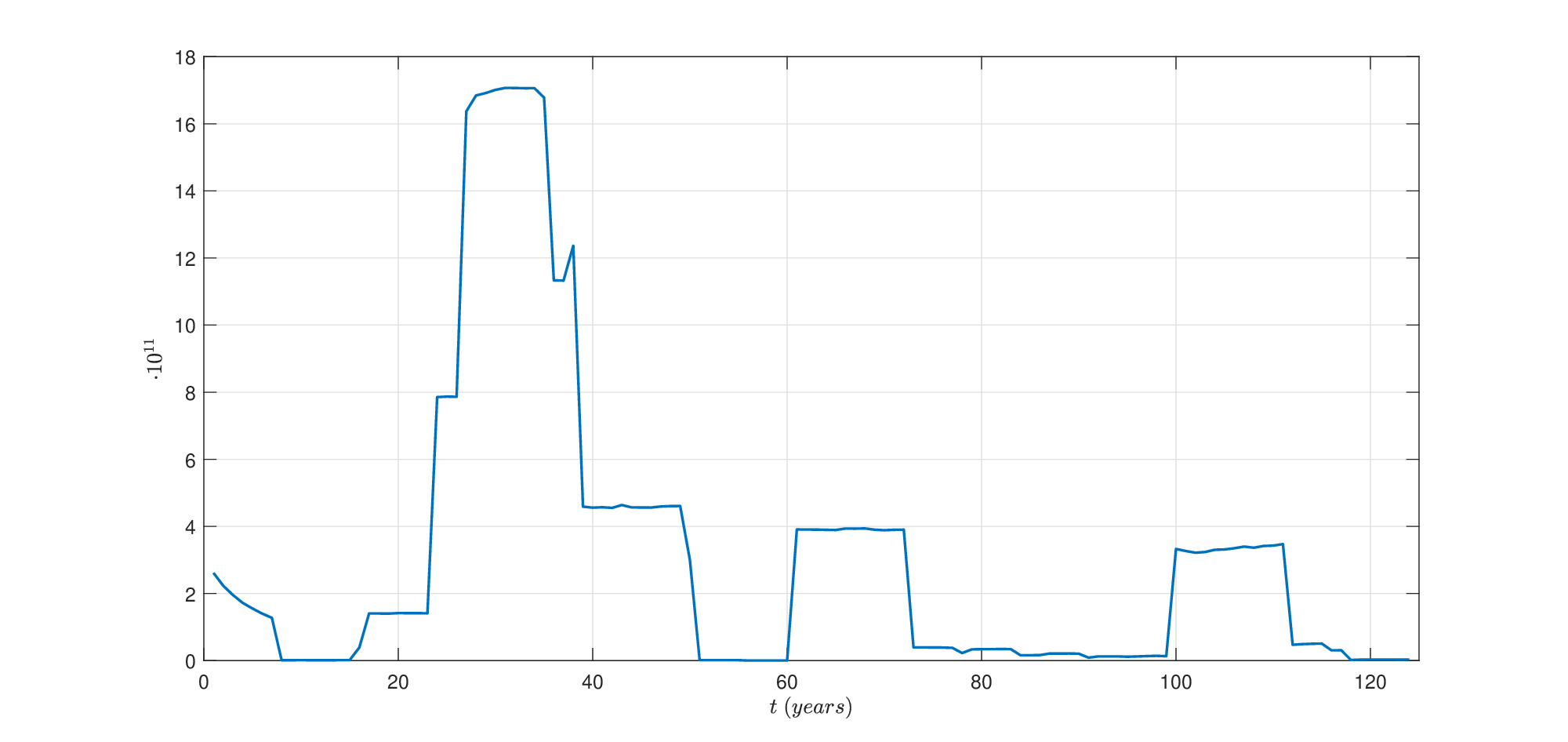}
	\caption{Moving variance $(V_u)_{u\in[1,N]}$ of total affected, with $N=124$ years (1898-2022).}
		\label{F:data}
\end{figure}

\section{Stochastic modeling} \label{Sec:Model}

Throughout the paper, we will make use of the following notation. We denote $\R_0^+:=[0,+\infty)$, $\R^+:=(0,+\infty)$ and by $C(\R_0^+;I)$ the space of continuous functions $f:\R_0^+ \to I\subset \R$. Furthermore, $\mathcal{M}(\R_0^+;I)$ will be the space of Borel-measurable functions $f:\R_0^+ \to I \subset \R$ and $\mathcal{M}(\R_0^+):=\mathcal{M}(\R_0^+;\R)$. Since the notation $L^\infty_{\rm loc}(\R_0^+)$ is usually adopted to denote Lebesgue-measurable functions that are locally bounded, we underline the requirement of Borel-measurability by writing $f \in \mathcal{M}(\R_0^+) \cap L^\infty_{\rm loc}(\R_0^+)$. The same will hold for all the Lebesgue spaces. By $\mathbf{1}_{I}$ we denote the indicator function of $I \subseteq \R$. We fix a probability space $(\Omega, \Sigma, \bP)$ supporting all the involved processes and random variables. By $W:=\{W_t,\ t \ge 0\}$ we denote a standard Brownian motion on it. For any continuous semimartingale $X=\{X_t,\ t \ge 0\}$ we denote by $[X,X]_t$ the quadratic variation process.

\subsection{The time-inhomogeneous skew Brownian motion and related stochastic differential equations}
We recall here the definition of time-inhomogeneous skew Brownian motion.

\begin{Definition}
	Let $\alpha \in \mathcal{M}(\R_0^+;[0,1])$ and $x \in \R$. A time-inhomogeneous skew Brownian motion (starting from $x$, with shape parameter $\alpha$) is a semimartingale $X^\alpha$ satisfying the following stochastic differential equation (SDE)
	\begin{equation}\label{eq:SDEskew}
		X^\alpha_t=x+W_t+\int_0^t(2\alpha(s)-1)dL_s(X^\alpha)
	\end{equation} 
	where $L(X^\alpha):=\{L_t(X^\alpha), \ t \ge 0\}$ is the symmetric local time in $0$ of $X^\alpha$, i.e.
	\begin{equation*}
		L_t(X^\alpha):=\lim_{\varepsilon \to 0}\frac{1}{2\varepsilon}\int_0^t\mathbf{1}_{[-\varepsilon,\varepsilon]}(X^\alpha_s)d[X^\alpha,X^\alpha]_s
	\end{equation*}
\end{Definition}

In \cite{weinryb1983etude}, the author proved pathwise uniqueness of any solution of 
\begin{equation}\label{eq:SDEskewplus}
	{}^+X^\alpha_t=x+W_t+\int_0^t(2\alpha(s)-1)dL^+_s({}^+X^\alpha)
\end{equation} 
where $L^+_s({}^+X^\alpha)$ is the right local time in $0$ of ${}^+X^\alpha$, i.e.
\begin{equation*}
	L^+_t({}^{+}X^\alpha):=\lim_{\varepsilon \to 0}\frac{1}{\varepsilon}\int_0^t\mathbf{1}_{[0,\varepsilon]}({}^{+}X^\alpha_s)d[{}^{+}X^\alpha,{}^{+}X^\alpha]_s
\end{equation*}
under the assumption $\alpha \in \mathcal{M}\left(\R_0^+; \left(-\infty,\frac{3}{4}\right]\right)$ (see also \cite[Exercise VI.2.24]{revuz2013continuous} for further details). However, as observed in \cite[Remark 2.3]{etore2012existence} (see also \cite[Therorem 1.1]{bouhadou2013time}), one has, with the same strategy, pathwise uniqueness of the solution of \eqref{eq:SDEskew} for any $\alpha \in \mathcal{M}(\R_0^+;[0,1])$. The actual existence and uniqueness result for strong solutions of \eqref{eq:SDEskew} is eventually shown in \cite[Theorem 2.13]{etore2012existence}. Recall that a further characterization of local time has been given in \cite[Theorem 22.1]{kallenberg1997foundations}, which leads to the a.e. inequality $L_t(X^\alpha) \le t$.
\begin{remark}
	In the special case $x=0$ and $\alpha(s)\equiv \alpha \in [0,1]$, we get the It\^o-McKean skew Brownian motion with shape parameter $\alpha$ (see \cite{harrison1981skew}). In particular, if we consider two independent Brownian motions $W^{(1)},W^{(2)}$, then the process
	\begin{equation}\label{eq:AzzalinisBm}
		X^\alpha_t=2\sqrt{\alpha(1-\alpha)}W^{(1)}_t+(2\alpha-1)|W^{(2)}_t|
	\end{equation} 
	is a It\^{o}-McKean skew Brownian motion (see \cite[Proposition 2.1]{corns2007skew}). The process \eqref{eq:AzzalinisBm} is usually referred as the Azzalini construction of the skew Brownian motion, as it can be easily shown to admit density belonging to the class of skew normal distributions introduced in \cite{azzalini1985class}. It is interesting to observe that if $\alpha=\frac{1}{2}$ we get the standard Brownian motion, if $\alpha=1$ we obtain the reflected Brownian motion and if $\alpha=0$ we get a negatively reflected Brownian motion. Hence, in some sense $X^\alpha$ interpolates between the negatively ad the positively reflected Brownian motion.
\end{remark}

Let us also recall that, as proved in \cite[Section 5]{etore2012existence}, $X^\alpha$ is a time-inhomogeneous Markov process whose transition density function is given by
\begin{align}\label{eq:conddist}
\begin{split}
	p_\alpha(t,y\mid &s,x):=\P(X^\alpha_t \in dy \mid X^\alpha_s=x)\\
	&=\int_0^{t-s}\frac{|y|(1+(2\alpha(u+s)-1){\sf sign}(y))}{2\pi(t-s-u)^{\frac{3}{2}}\sqrt{u}}\exp\left(-\frac{y^2}{2(t-s-u)}-\frac{x^2}{2u}\right)\,du\\
	&+\frac{1}{\sqrt{2\pi(t-s)}}\left[\exp\left(-\frac{(y-x)^2}{2(t-s)}\right)-\exp\left(-\frac{(y+x)^2}{2(t-s)}\right)\right]\mathbf{1}_{\R^+}(xy),
\end{split}
\end{align}
where $x,y \in \R$ and $0 \le s<t$. The joint density of $(X_t^\alpha,L_t(X^\alpha))$ is given in \cite[Theorem 2.14]{etore2012existence}. If we define
\begin{equation*}
	A^\alpha_t:=x+\int_0^t (2\alpha(s)-1)dL_s(X^\alpha),
\end{equation*}
as a direct consequence of \cite[Proposition IV.2.10]{revuz2013continuous} we know that $A^\alpha$ has locally bounded variation and then the decomposition of $X^\alpha$ as a semimartingale is given by
\begin{equation*}
	X^\alpha_t=W_t+A^\alpha_t,
\end{equation*}
which in turn implies $[X^\alpha,X^\alpha]_t=t$. Once the decomposition of $X^\alpha$ as a semimartingale is clear, one can easily define the space of stochastic integrands with respect to the time-inhomogeneous skew Brownian motion, according to \cite[Definitions IV.2.6, IV.2.8 and IV.2.9]{revuz2013continuous}.
\begin{Definition}
Let $Y=\{Y_t, t \ge 0\}$ be a progressively measurable stochastic process. We say that $Y$ is locally bounded if there exist an increasing sequence of stopping times $\{T_n\}_{n \in \NN}$ such that $T_n \to +\infty$ a.s. and a sequence of positive constants $\{C_n\}_{n \in \NN}$ such that for any $n \in \NN$ and $t \ge 0$ it holds
	\begin{equation*}
		|Y_{t \wedge T_n}|\le C_n \mbox{ a.s.}
	\end{equation*}
	We say that a process $Y=\{Y_t, t \in [0,T]\} \in M^2_T(\Omega)$ if $Y$ is progressively measurable, a.s. bounded and 
	\begin{equation*}
		\int_0^T \E[Y_t^2]dt<\infty.
	\end{equation*}
	For any process $Y \in M^2_T(\Omega)$ we can define
	\begin{equation*}
		\int_0^T Y_tdX^\alpha_t=\int_0^TY_tdW_t+\int_0^T(2\alpha(t)-1)Y_tdL_t(X^\alpha).
	\end{equation*}
	Finally, we say that $Y=\{Y_t, t \ge 0\} \in M^2_{\rm loc}(\Omega)$ if $Y$ is locally bounded and its restriction on $[0,T]$ belongs to $M^2_T(\Omega)$ for any $T>0$. 
\end{Definition}
Once this is done, it is not difficult to define stochastic differential equations driven by $X^\alpha$.
\begin{Definition}
	Let $T>0$, $b,\sigma:[0,T]\times \R \times \Omega \to \R$ be measurable functions, $\overline{Y}$ be a random variable and $\alpha \in \mathcal{M}(\R_0^+;[0,1])$. We say that a process $Y \in M^2_T(\Omega)$ is a strong solution of the stochastic differential equation (SDE)
	\begin{equation}\label{SDEproto}
		\begin{cases}
			dY_t=b(t,Y_t)dt+\sigma(t,Y_t)dX^\alpha_t & t \in [0,T]\\
			Y_0=\overline{Y}
		\end{cases}
	\end{equation}
	if
	\begin{itemize}
		\item[$(i)$] The stochastic process $t \in [0,T] \mapsto b(t,Y_t) \in \R$ belongs to $L^1[0,T]$ a.s.
		\item[$(ii)$] The stochastic process $t \in [0,T] \mapsto \sigma(t,Y_t) \in \R$ belongs to $M^2_T(\Omega)$
		\item[$(iii)$] It holds
		\begin{equation*}
			Y_t=\overline{Y}+\int_0^tb(s,Y_s)ds+\int_0^t\sigma(s,Y_s)dX^\alpha_s.
		\end{equation*}
	\end{itemize}
	We say that pathwise uniqueness holds for \eqref{SDEproto} if for any two strong solutions $Y^1,Y^2$ we have
	\begin{equation*}
		\P(Y^1(t)=Y^2(t), \ \forall t \in [0,T])=1.
	\end{equation*}
	If $T=\infty$, we ask $Y \in M^2_{\rm loc}(\Omega)$.
\end{Definition}

\subsection{The time-inhomogeneous geometric-skew Brownian motion (GSBM)}
Now let $\mu,\sigma \in \mathcal{M}(\R_0^+)$ such that $\mu,\sigma \in L^\infty_{\rm loc}(\R_0^+)$ and consider a random variable $\overline{G}^\alpha \in L^p(\Omega)$ for some $p>2$ and such that $\bP(\overline{G}^\alpha>0)=1$. We want to focus on the following linear SDE
\begin{equation}\label{eq:geomSDE}
	\begin{cases}
	dG^\alpha_t=\mu(t)G^\alpha_tdt+\sigma(t)G^\alpha_tdX^\alpha_t & t \ge 0\\
	G^\alpha_0=\overline{G}^\alpha
	\end{cases}
\end{equation}
Namely, we now prove the following existence and uniqueness theorem.

\begin{Theorem}
	The pathwise unique strong solution of \eqref{eq:geomSDE} is given by
	\begin{equation}\label{eq:geomskew}
		G_t^\alpha=\overline{G}^\alpha\exp\left(\int_0^t\left(\mu(s)-\frac{\sigma^2(s)}{2}\right)ds+\int_0^t\sigma(s)dX^\alpha_s\right)
	\end{equation}
\end{Theorem}

\begin{proof}
    It is not difficult to check, since $\overline{G}^\alpha \in L^p(\Omega)$ for $p>2$, that $G^\alpha \in M^2_{\rm loc}(\Omega)$. Indeed, first observe that $G^\alpha$ is clearly locally bounded. Furthermore, recalling that $L_t(X^\alpha_t) \le t$ and using Young's inequality, we have for any $T>0$ and $t \in [0,T]$
\begin{align*}
    \E&\left[\left(\overline{G}^\alpha\exp\left(\int_0^t\left(\mu(s)-\frac{\sigma^2(s)}{2}\right)ds+\int_0^t\sigma(s)dX^\alpha_s\right)\right)^2\right]\\
    &\le \frac{2}{p}\E\left[\left(\overline{G}^\alpha\right)^p\right]+ \frac{p-2}{p}\E\left[\exp\left(\frac{p}{p-2}\int_0^t\left(\mu(s)-\frac{\sigma^2(s)}{2}\right)ds+\frac{p}{p-2}\int_0^t\sigma(s)dX^\alpha_s\right)\right]\\
    &\le \frac{2}{p}\E\left[\left(\overline{G}^\alpha\right)^p\right]+ \frac{p-2}{p}\exp\left(\frac{p}{p-2}\int_0^t\left(\mu(s)-\frac{\sigma^2(s)}{2}\right)ds\right)\\
    &\quad \times \E\left[\exp\left(\frac{p}{p-2}\int_0^t\sigma(s)(2\alpha(s)-1)dL_s(X^\alpha)+\frac{p}{p-2}\int_0^t\sigma(s)dW_s\right)\right]\\
    &\le \frac{2}{p}\E\left[\left(\overline{G}^\alpha\right)^p\right]+ \frac{p-2}{p}\exp\left(\frac{p}{p-2}\int_0^t\left(\mu(s)-\frac{\sigma^2(s)}{2}\right)ds\right)\\
    &\quad \times \E\left[\exp\left(\frac{p}{p-2}\left\Vert\sigma \right\Vert_{L^\infty(0,T)}(2\left\Vert\alpha\right\Vert_{L^\infty(0,T)}+1)L_t+\frac{p}{p-2}\int_0^t\sigma(s)dW_s\right)\right]\\
    &\le \frac{2}{p}\E\left[\left(\overline{G}^\alpha\right)^p\right]+ \frac{p-2}{p}\exp\left(\frac{p}{p-2}\int_0^t\left(\mu(s)-\frac{\sigma^2(s)}{2}\right)ds\right)\\
    &\quad \times\exp\left(\frac{p}{p-2}\left\Vert\sigma \right\Vert_{L^\infty(0,T)}(2\left\Vert\alpha \right\Vert_{L^\infty(0,T)}+1)T\right) \E\left[\exp\left(\frac{p}{p-2}\int_0^t\sigma(s)dW_s\right)\right].
 \end{align*}
    Next, we recall that $\int_0^t \sigma(s)dW_s$ is a centered Gaussian random variable with variance $\int_0^t \sigma^2(s)ds$, hence we can evaluate its moment generating function obtaining
    \begin{align*}
    \E&\left[\left(\overline{G}^\alpha\exp\left(\int_0^t\left(\mu(s)-\frac{\sigma^2(s)}{2}\right)ds+\int_0^t\sigma(s)dX^\alpha_s\right)\right)^2\right]\\
    &\le \frac{2}{p}\E\left[\left(\overline{G}^\alpha\right)^p\right]+ \frac{p-2}{p}\exp\left(\frac{p}{p-2}\int_0^t\left(\mu(s)-\frac{\sigma^2(s)}{2}\right)ds\right)\\
    &\quad \times \exp\left(\frac{p}{p-2}\left\Vert\sigma \right\Vert_{L^\infty(0,T)}(2\left\Vert\alpha \right\Vert_{L^\infty(0,T)}+1)T+\frac{1}{2}\left(\frac{p}{p-2}\right)^2\int_0^t\sigma^2(s)ds\right),
    \end{align*}
    where the right-hand side is clearly integrable in $[0,T]$. Since $T>0$ is arbitrary, we get that $G^\alpha \in M^2_{\rm loc}(\Omega)$.
    
    Let us now show that \eqref{eq:geomskew} solves \eqref{eq:geomSDE}. Upon substituting $G^\alpha_t$ with $\frac{G^\alpha_t}{\overline{G}^\alpha}$, we can assume $\overline{G}^\alpha \equiv 1$. Define the process
	\begin{equation*}
		Y^\alpha_t=\int_0^t\sigma(s)dX^\alpha_s.
	\end{equation*}
	This is well-defined since $\sigma \in L^\infty_{\rm loc}(\R_0^+)\subset M^2_t(\Omega)$ for any $t \ge 0$. Furthermore, it is a semimartingale. Indeed, if we define
	\begin{equation*}
		A^{\alpha,\sigma}_t=\int_0^t\sigma(s)(2\alpha(s)-1)dL_s(X^\alpha),
	\end{equation*}
	by \cite[Proposition IV.2.10]{revuz2013continuous} we know that $A^{\alpha,\sigma}$ has locally bounded variation and then $Y^\alpha$ can be decomposed as
	\begin{equation*}
		Y^\alpha_t=\int_0^t\sigma(s)dW_s+A^{\alpha,\sigma}_t.
	\end{equation*}
	Furthermore, it is clear that
	\begin{equation*}
		dY^\alpha_t=\sigma(t)dW_t+\sigma(t)(2\alpha(t)-1)dL_t(X^\alpha)=\sigma(t)dX^\alpha_t
	\end{equation*}
	and
	\begin{equation*}
		d[Y^\alpha,Y^\alpha]_t=\sigma^2(t)dt.
	\end{equation*}
	Next, define
	\begin{equation*}
		Z^\alpha_t=Y^\alpha_t+\int_0^t\left(\mu(s)-\frac{\sigma^2(s)}{2}\right)ds
	\end{equation*}
	which is still a semi-martingale, since we are adding to $Y^\alpha$ a function of locally bounded variation (precisely, an absolutely continuous function), with $d[Z^\alpha,Z^\alpha]_t=\sigma^2(t)dt$. By definition, we have $G^\alpha_t=\exp(Z^\alpha_t)$, hence, by It\^{o}'s formula (see \cite[Theorem IV.3.3]{revuz2013continuous})
	\begin{align*}
		dG^\alpha_t&=G^\alpha_tdZ^\alpha_t+\frac{1}{2}G^\alpha_td[Z^\alpha,Z^\alpha]_t\\
		&=G^\alpha_t\left(\mu(t)-\frac{\sigma^2(t)}{2}\right)dt+G^\alpha_t\sigma(t)dX^\alpha_t+\frac{G^\alpha_t\sigma^2(t)}{2}dt\\
		&=\mu(t)G^\alpha_tdt+\sigma(t)G^\alpha_tdX^\alpha_t.
	\end{align*}
	Now let us show that $G^\alpha$ is the unique solution. To do this, let $Y$ be any other strong solution and fix $T>0$. Since $Y \in M_{\rm loc}^2(\Omega)$, it is locally bounded and there exist a sequence $T_n \to +\infty$ of stopping times and a sequence $C_n>0$ of constants such that $|Y_{t \wedge T_n}| \le C_n$. Fix $T>0$, $\omega \in \Omega$ and let $n=n(\omega)$ such that $T_n(\omega)>T$. It is clear that the function $t \in [0,T]\mapsto \int_0^t \mu(s)Y_s(\omega)ds$ is continuous. Furthermore for $s\le t \le T$
	\begin{align*}
	&\left|\int_0^t \sigma(\tau)(2\alpha(\tau)-1)Y_\tau(\omega)dL_\tau(X^\alpha)(\omega)-\int_0^s \sigma(\tau)(2\alpha(\tau)-1)Y_\tau(\omega)dL_\tau(X^\alpha)(\omega)\right|\\
	&\qquad =\left|\int_s^t \sigma(\tau)(2\alpha(\tau)-1)Y_\tau(\omega)dL_\tau(X^\alpha)(\omega)\right|\\
	&\qquad \le \left\Vert \sigma \right\Vert_{L^\infty[0,T]}(2\left\Vert \alpha \right\Vert_{L^\infty[0,T]}+1)C_n(L_t(X^\alpha)(\omega)-L_s(X^\alpha)(\omega)).
	\end{align*}
	Since $L_t(X^\alpha)$ is the local time of a continuous semimartingale, we can consider a version that is continuous in $t$ and then the previous inequality guarantees that $t \in [0,T] \mapsto \int_0^t \sigma(\tau)(2\alpha(\tau)-1)Y_\tau dL_\tau(X^\alpha)$ is a.s. continuous. Finally, observe that since $\sigma \in L^\infty_{\rm loc}[0,T]$ and $Y \in M^2_{\rm loc}(\Omega)$,  the process $t \in [0,T] \mapsto \int_0^t \sigma(\tau)Y_\tau dW_\tau$ is a continuous martingale (see the discussion in \cite[Section 3.2]{karatzas1991brownian}). Hence $Y$ is a.s. continuous. Let $T_0:=\inf\{t \ge 0: \ Y_t=0\}$, where, clearly, $T_0>0$ a.s. Thus, the process $Z_t=\log\left(\frac{Y_t}{\overline{G}^\alpha}\right)$ is well-defined for $t < T_0$ and $Z_t \to -\infty$ as $t \to T_0$ a.s. Now assume by contradiction that $\bP(T_0<\infty)>0$ and consider $\omega \in \{T_0<\infty\}$. Then, by It\^{o}'s formula (omitting the dependence on $\omega$ for the ease of the reader),
	\begin{equation*}
		Z_{t \wedge T_0}=\int_0^{t \wedge T_0}\left(\mu(s)-\frac{\sigma^2(s)}{2}\right)ds+\int_0^{t \wedge T_0}\sigma(s)dX^\alpha_s.
	\end{equation*}
	and taking the limit as $t \uparrow T_0$ we have
	\begin{equation}
		\lim_{t \to T_0}Z_t=\int_0^{T_0}\left(\mu(s)-\frac{\sigma^2(s)}{2}\right)ds+\int_0^{T_0}\sigma(s)dX^\alpha_s
	\end{equation}
	which is absurd since the left-hand side is finite. Hence $\bP(T_0=\infty)=1$ and we have for any $t \ge 0$
	\begin{equation*}
		Z_{t}=\int_0^{t}\left(\mu(s)-\frac{\sigma^2(s)}{2}\right)ds+\int_0^{t}\sigma(s)dX^\alpha_s
	\end{equation*}
	which in turn implies that $Y_t=G^\alpha_t$ for any $t \ge 0$.
\end{proof}
\subsection{Approximation of $G^\alpha_t$ with piecewise constant parameters}
From now on, without loss of generality, we assume that $X^\alpha_0=0$.

For the forecasting procedure, we will approximate the functional parameters $\mu,\sigma,\alpha$ of the process with piecewise constants functions. To do this, let us denote by $\mathcal{A}$ the class of functions $f \in \mathcal{M}(\R_0^+)\cap L^{\infty}_{\rm loc}(\R_0^+)$ with the following property: for any $T>0$ and any sequence of partitions $\Pi_n:t^n_0=0<t^n_1<\cdots<t^n_n=T$ with ${\sf diam}(\Pi_n):=\max_{i=1,\dots,n}|t^n_i-t^n_{i-1}| \to 0$ as $n \to +\infty$, there exist two sequences of c\`{a}dl\`{a}g functions $\underline{f}_n,\overline{f}_n:[0,T] \to \R$, $n=1,2,\dots$ that are constant on each interval $[t^n_i,t^n_{i+1})$, $i=0,\dots,n-1$, and satisfy $\underline{f}_n(t)\le f(t) \le \overline{f}_n(t)$ and $\lim_{n \to +\infty}\underline{f}_n(t)=\lim_{n \to +\infty}\overline{f}_n(t)=f(t)$ for all $t \in [0,T]$. This condition has been considered for instance in \cite{etore2012existence} only in the interval $[0,1]$, where it is called Condition $\mathcal{H}$, and it is clearly satisfied by any continuous function. If $\alpha \in \mathcal{A}$, then, by a simple adaptation of the proof of \cite[Theorem 7.4]{etore2012existence} to the case of the interval $[0,T]$ (in place of just $[0,1]$) and using explicitly \cite[Equation (7.3)]{etore2012existence}, one has for any $T>0$ and any $\underline{\alpha}_n,\overline{\alpha}_n$ as before
\begin{equation}\label{eq:monotoneapprox}
	\lim_{n \to +\infty}\E\left[\sup_{t \in [0,T]}\left|X^{\underline{\alpha}_n}_t-X^{\alpha}_t\right|\right]=\lim_{n \to +\infty}\E\left[\sup_{t \in [0,T]}\left|X^{\overline{\alpha}_n}_t-X^{\alpha}_t\right|\right]=0.
\end{equation}
We can extend the previous result to any piecewise constant approximation of $\alpha$.
\begin{prop}
Assume $\alpha \in \mathcal{A}$. For any sequence of partitions $\Pi_n:t^n_0=0<t^n_1<\cdots<t^n_n=T$ with ${\sf diam}(\Pi_n) \to 0$ as $n \to +\infty$ and any sequence of c\`{a}dl\`{a}g functions $\alpha_n:[0,T] \to [0,1]$, where $\alpha_n(t)$ is constant on $[t_i^n,t_{i+1}^n)$, $i=0,\dots,n-1$ and $\lim_{n \to +\infty}\alpha_n(t)=\alpha(t)$ for any $t \in [0,T]$, it holds
\begin{equation}\label{eq:limitalpha}
	\lim_{n \to +\infty}\E\left[\sup_{t \in [0,T]}\left|X^{\alpha_n}_t-X^{\alpha}_t\right|\right]=0.
\end{equation}    
\end{prop}
\begin{proof}
Consider a sequence of partitions $\Pi_n$ and a sequence of c\`{a}dl\`{a}g functions $\alpha_n:[0,T] \to [0,1]$ as in the statement. Since $\alpha \in \mathcal{A}$, we can also consider two sequences of c\`{a}dl\`{a}g functions $\underline{\alpha}_n,\overline{\alpha}_n:[0,T] \to [0,1]$ that are constant on each interval $[t_i^n,t_i^{n+1})$, $\underline{\alpha}_n(t) \le \alpha(t) \le \overline{\alpha}_n(t)$ and $\lim_{n \to +\infty}\underline{\alpha}_n(t)=\lim_{n \to +\infty}\overline{\alpha}_n(t)=\alpha(t)$ for any $t \in [0,T]$. For each $n \in \mathbb{N}$ and $t \in [0,T]$ we define $\underline{\underline{\alpha_n}}(t)=\min\{\alpha_n(t),\underline{\alpha}_n(t)\}$ and $\overline{\overline{\alpha_n}}(t)=\max\{\alpha_n(t),\overline{\alpha}_n(t)\}$. By definition, $\underline{\underline{\alpha_n}},\overline{\overline{\alpha_n}}$ are c\`{a}dl\`{a}g, piecewise constant on $[t_i^n,t_i^{n+1})$ and $\underline{\underline{\alpha_n}}(t) \le \alpha(t) \le \overline{\overline{\alpha_n}}(t)$, $\underline{\underline{\alpha_n}}(t) \le \alpha_n(t) \le \overline{\overline{\alpha_n}}(t)$ and $\lim_{n \to +\infty}\underline{\underline{\alpha_n}}(t)=\lim_{n \to +\infty}\overline{\overline{\alpha_n}}(t)=\alpha(t)$ for any $t \in [0,T]$, hence, by \eqref{eq:monotoneapprox},
\begin{equation*}
	\lim_{n \to +\infty}\E\left[\sup_{t \in [0,T]}\left|X^{\underline{\underline{\alpha_n}}}_t-X^{\alpha}_t\right|\right]=\lim_{n \to +\infty}\E\left[\sup_{t \in [0,T]}\left|X^{\overline{\overline{\alpha_n}}}_t-X^{\alpha}_t\right|\right]=0.
\end{equation*}
By the comparison theorem for solutions of SDEs involving local time (see \cite[Theorem 3.4]{le2006one}), we know that a.e. and for any $t \in [0,T]$
\begin{equation*}
X_t^{\underline{\underline{\alpha_n}}}\le X_t^{\alpha_n} \le X_t^{\overline{\overline{\alpha_n}}}.
\end{equation*}
and
\begin{equation*}
X_t^{\underline{\underline{\alpha_n}}}\le X_t^{\alpha} \le X_t^{\overline{\overline{\alpha_n}}}.
\end{equation*}
In particular, we get, a.e. and for any $t \in [0,T]$
\begin{equation*}
|X_t^{\alpha_n}-X_t^{\alpha}| \le |X_t^{\underline{\underline{\alpha_n}}}-X_t^{\alpha}|+ |X_t^{\overline{\overline{\alpha_n}}}-X_t^{\alpha}|.
\end{equation*}
Taking the supremum over $[0,T]$, the expectation and then the limit as $n \to \infty$ we get the desired result.
\end{proof}

After the preceding conditions and results have been established, we can provide the following approximation result for $G_t^\alpha$.

\begin{Theorem}\label{th1}
Assume $\alpha \in \A$, $\mu,\sigma \in \mathcal{M}(\R_0^+)\cap L^{\infty}_{\rm loc}(\R_0^+)$ and fix $T>0$. 
Consider a sequence of partitions $\Pi_n:t^n_0=0<t^n_1<\cdots<t^n_n=T$ with ${\sf diam}(\Pi_n) \to 0$ as $n \to +\infty$. Let also $\mu_n,\sigma_n:[0,T] \to \R$ and $\alpha_n:[0,T] \to [0,1]$ be sequences of c\`{a}dl\`{a}g functions such that

\begin{itemize}
	\item[$(i)$] $\alpha_n,\mu_n,\sigma_n$ are constant on each interval $[t_i^n,t^n_{i+1})$, $i=0,\dots,n-1$;
	\item[$(ii)$] $\alpha_n(t) \to \alpha(t)$, $\mu_n(t) \to \mu(t)$ and $\sigma_n(t) \to \sigma(t)$ for any $t \in [0,T]$;
	\item[$(iii)$] There exists a constant $M$ such that
	\begin{equation*}
		\sup_{n \in \NN}\sup_{t \in [0,T]}(|\mu_n(t)|+|\sigma_n(t)|) \le M;
	\end{equation*} 
	\item[$(iv)$] There exists a constant $V$ such that
	\begin{equation*}
		\sup_{n \in \NN}\sum_{i=0}^{n}|\sigma_n(t_{i+1})-\sigma_n(t_{i})| \le V.
	\end{equation*}
\end{itemize}
Fix $\overline{G}^\alpha$ and denote by $G^{\alpha,\mu,\sigma}$ the solution of \eqref{eq:geomSDE} with parameters $\alpha,\mu,\sigma$. Then
\begin{equation}\label{eq:geommeanconv}
	\lim_{n \to +\infty}\E\left[\sup_{t \in [0,T]}\left|\log\left(\frac{G^{\alpha,\mu,\sigma}_t}{G^{\alpha_n,\mu_n,\sigma_n}_t}\right)\right|\right]=0.
\end{equation}
\end{Theorem}
\begin{proof}
	Let us first observe that
	\begin{equation*}
		\left|\log\left(\frac{G^{\alpha,\mu,\sigma}_t}{G^{\alpha_n,\mu_n,\sigma_n}_t}\right)\right| \le \left|\log\left(\frac{G^{\alpha,\mu,\sigma}_t}{G^{\alpha,\mu_n,\sigma_n}_t}\right)\right|+\left|\log\left(\frac{G^{\alpha,\mu_n,\sigma_n}_t}{G^{\alpha_n,\mu_n,\sigma_n}_t}\right)\right|.
	\end{equation*}
	To estimate the first logarithm, notice that
	\begin{equation*}
		\left|\log\left(\frac{G^{\alpha,\mu,\sigma}_t}{\widetilde{G}^{\alpha,\mu_n,\sigma_n}_t}\right)\right|\le \int_0^t \left|\mu(s)-\mu_n(s)-\frac{\sigma^2(s)-\sigma^2_n(s)}{2}\right|ds+\left|\int_0^t \left(\sigma_n(s)-\sigma(s)\right)dX_s^\alpha\right|.
	\end{equation*}
	For the first integral, we simply get
	\begin{multline*}
		\E\left[\sup_{t \in [0,T]}\int_0^t \left|\mu(s)-\mu_n(s)-\frac{\sigma^2(s)-\sigma^2_n(s)}{2}\right|ds\right]\\
        =\E\left[\int_0^T \left|\mu(s)-\mu_n(s)-\frac{\sigma^2(s)-\sigma^2_n(s)}{2}\right|ds\right] \to 0
	\end{multline*}
	as $n \to \infty$ by the dominated convergence theorem. To handle the second summand, observe that
	\begin{equation*}
		\int_0^t \left(\sigma_n(s)-\sigma(s)\right)dX_s^\alpha=\int_0^t \left(\sigma_n(s)-\sigma(s)\right)(2\alpha(t)-1)dL_t(X^\alpha)+\int_0^t \left(\sigma_n(s)-\sigma(s)\right)dW_s
	\end{equation*}
	and then
	\begin{equation*}
		\left|\int_0^t \left(\sigma_n(s)-\sigma(s)\right)dX_s^\alpha\right|\le \int_0^t \left|\sigma_n(s)-\sigma(s)\right||2\alpha(t)-1|dL_t(X^\alpha)+\left|\int_0^t \left(\sigma_n(s)-\sigma(s)\right)dW_s\right|.
	\end{equation*}
	Taking the supremum and the expectation we have
	\begin{align*}
		\E\left[\sup_{t \in [0,T]}\left|\int_0^t \left(\sigma_n(s)-\sigma(s)\right)dX_s^\alpha\right|\right] &\le \E\left[\int_0^T \left|\sigma_n(s)-\sigma(s)\right||2\alpha(t)-1|dL_t(X^\alpha)\right]\\
		&+\E\left[\sup_{t \in [0,T]}\left|\int_0^t \left(\sigma_n(s)-\sigma(s)\right)dW_s\right|\right].
	\end{align*}
	On the one hand, we have $\left|\sigma_n(s)-\sigma(s)\right||2\alpha(t)-1| \le 2M$, where
	\begin{equation*}
		\E\left[\int_0^T 2MdL_t(X^\alpha)\right] \le 2M \E\left[L_T(X^\alpha)\right] \le 2MT<\infty.
	\end{equation*}
	Hence, by the dominated convergence theorem,
	\begin{equation*}
		\lim_{n \to \infty}\E\left[\int_0^T \left|\sigma_n(s)-\sigma(s)\right||2\alpha(t)-1|dL_t(X^\alpha)\right]=0.
	\end{equation*}
	On the other hand, we have, by Doob's maximal inequality \cite[Theorem II.1.7]{revuz2013continuous}
	\begin{multline*}
		\E\left[\sup_{t \in [0,T]}\left|\int_0^t \left(\sigma_n(s)-\sigma(s)\right)dW_s\right|\right] \le \sqrt{\E\left[\sup_{t \in [0,T]}\left|\int_0^t \left(\sigma_n(s)-\sigma(s)\right)dW_s\right|^2\right]} \\
		\le 2\sqrt{\E\left[\int_0^T \left|\sigma_n(s)-\sigma(s)\right|^2_s\right]} \to 0,
	\end{multline*}
		where the limit holds by dominated convergence. This proves that
		\begin{equation*}
			\lim_{n \to +\infty}\E\left[\sup_{t \in [0,T]}\log\left(\frac{G^{\alpha,\mu,\sigma}_t}{\widetilde{G}^{\alpha,\mu_n,\sigma_n}_t}\right)\right]=0.
		\end{equation*}
	Now let us work with the second logarithm. Let $Y^n_t=X^{\alpha}_t-X^{\alpha_n}_t$ and $i_{\star}^n(t)=\min\{i=0,\dots, n-1: \ t_i^n \ge t\}$. We have
	\begin{align*}
		\left|\log\left(\frac{G^{\alpha,\mu_n,\sigma_n}_t}{G^{\alpha_n,\mu_n,\sigma_n}_t}\right)\right|&=\left|\sum_{i=1}^{i_\star^n(t)-1}\sigma_n(t_{i-1}^n)(Y^n_{t_{i}^n}-Y^n_{t_{i-1}^n})+\sigma_n(t_{i^n_{\star}(t)-1}^n)(Y^n_{t}-Y^n_{t_{i_{\star}^n(t)-1}^n})\right|\\
		&=\left|\sum_{i=1}^{i_\star^n(t)-1}\left(\sigma_n(t_{i-1}^n)-\sigma_n(t_{i}^n)\right)Y^n_{t_{i}^n}+\sigma_n(t^n_{i^n_\star(t)-1})Y^n_t\right|\\
		&\le \left(\sum_{i=1}^{i_\star^n(t)-1}\left|\sigma_n(t_{i-1}^n)-\sigma_n(t_{i}^n)\right|+|\sigma_n(t^n_{i^n_\star(t)-1})|\right)\sup_{t \in [0,T]}|Y^n_t|\\
		&\le (V+M)\sup_{t \in [0,T]}|Y^n_t|.
	\end{align*}
	Taking the supremum, the expectation and the limit as $n \to +\infty$ and using \eqref{eq:limitalpha} we get
	\begin{equation*}
		\lim_{n \to +\infty}\E\left[\left|\log\left(\frac{G^{\alpha,\mu_n,\sigma_n}_t}{G^{\alpha_n,\mu_n,\sigma_n}_t}\right)\right|\right] \le (V+M)\lim_{n \to +\infty}\E\left[\sup_{t \in [0,T]}|Y^n_t|\right]=0.
	\end{equation*} 
	This ends the proof.
\end{proof}
\begin{remark}
	Observe that condition $(iv)$ is satisfied if, for instance, $\sigma$ is a (c\`{a}dl\`{a}g) bounded variation function and $\sigma_n(t)=\sigma(t_{i-1}^n)$ for any $t \in [t_{i-1}^n,t_i^n)$, $i=1,\dots,n$.
	
	Furthermore, as a consequence of the uniform geometric mean convergence \eqref{eq:geommeanconv}, we know that $\frac{G_t^{\alpha_n,\mu_n,\sigma_n}}{G_t^{\alpha,\mu,\sigma}} \overset{\bP}{\to} 1$ and there exists a subsequence such that $\frac{G_t^{\alpha_{n_k},\mu_{n_k},\sigma_{n_k}}}{G_t^{\alpha,\mu,\sigma}} \overset{a.s.}{\to} 1$.
\end{remark}

Now, let us evaluate the conditional expectation of $G_t^\alpha$ given $G_s^\alpha$, where $t \ge s$, in case $\sigma$ is constant.
\begin{prop}
Let $s \le t$ and assume that $\sigma(s)\equiv \sigma$ and assume further that $\overline{G}^\alpha$ is degenerate. Then we have

\begin{align} \label{eq:condexp}
 \begin{split}
 &\E[G^{\alpha}_t \mid G^{\alpha}_s]\\
	&=G^\alpha_s \exp\left(\int_0^t\left(\mu(\tau)-\frac{\sigma^2}{2}\right)d\tau-\log\left(\frac{G_s^\alpha}{\overline{G}^\alpha}\right)\right)\\
	&\ \times \int_{\R}\int_0^{t-s}\frac{|y|(1+(2\alpha(u+s)-1){\sf sign}(y))}{2\pi(t-s-u)^{\frac{3}{2}}}\exp\left(-\frac{y^2}{2(t-s-u)}-\frac{(F(G^\alpha_s))^2}{2u}+\sigma y\right)du\, dy\\
	&\ +\frac{1}{\sqrt{2\pi(t-s)}}\int_{\R}\left[\exp\left(-\frac{(y-F(G^\alpha_s))^2}{2(t-s)}\right)-\exp\left(-\frac{(y+F(G^\alpha_s))^2}{2(t-s)}\right)\right]e^{\sigma y}\mathbf{1}_{\R^+}(F(G^\alpha_s)y)dy,
	\end{split}
\end{align}

where
\begin{equation}\label{eq:Fsigmamu}
	F(x)=\frac{1}{\sigma}\left(\log\left(\frac{x}{\overline{G}^\alpha}\right)-\int_0^s\left(\mu(\tau)-\frac{\sigma^2}{2}\right)d\tau\right)
\end{equation}
\end{prop}
\begin{proof}
We clearly have
\begin{equation*}
	G^\alpha_t=G^\alpha_s \exp\left(\int_s^t\left(\mu(\tau)-\frac{\sigma^2}{2}\right)d\tau+\sigma (X^\alpha_t-X^\alpha_s)\right).
\end{equation*}
Furthermore, observe that
\begin{equation*}
	X^\alpha_s=\frac{1}{\sigma}\left(\log\left(\frac{G_s^\alpha}{\overline{G}^\alpha}\right)-\int_0^s\left(\mu(\tau)-\frac{\sigma^2}{2}\right)d\tau\right)=F(G^\alpha_s).
\end{equation*}
Hence we obtain
\begin{equation}\label{eq:condexp1}
	\E[G^{\alpha}_t \mid G^{\alpha}_s]=G^\alpha_s \exp\left(\int_0^t\left(\mu(\tau)-\frac{\sigma^2}{2}\right)d\tau-\log\left(\frac{G_s^\alpha}{\overline{G}^\alpha}\right)\right)\E\left[e^{\sigma X^\alpha_t} \mid G^\alpha_s\right].
\end{equation}
To evaluate the remaining conditional expectation, we use \eqref{eq:conddist} so that
\begin{align*}
	&\E\left[e^{\sigma X^\alpha_t} \mid G^\alpha_s\right]=\int_{\R}e^{\sigma y}p_\alpha(t,y\mid s, F(G^\alpha_s))dy\\
	&=\int_{\R}\int_0^{t-s}\frac{|y|(1+(2\alpha(u+s)-1){\sf sign}(y))}{2\pi(t-s-u)^{\frac{3}{2}}}\exp\left(-\frac{y^2}{2(t-s-u)}-\frac{(F(G^\alpha_s))^2}{2u}+\sigma y\right)du\, dy\\
	&\quad +\frac{1}{\sqrt{2\pi(t-s)}}\int_{\R}\left[\exp\left(-\frac{(y-F(G^\alpha_s))^2}{2(t-s)}\right)-\exp\left(-\frac{(y+F(G^\alpha_s))^2}{2(t-s)}\right)\right]e^{\sigma y}\mathbf{1}_{\R^+}(F(G^\alpha_s)y)dy.
\end{align*}
Combining this with \eqref{eq:condexp1} we get the desired formula.
\end{proof}

\section{Empirical results} \label{Sec:Results}


To calibrate the functional parameters $\mu(t), \sigma(t), \alpha(t)$ for any $t \in [0, T]$, which are unknown a priori, and forecast our time series $V$, we follow the algorithm in Table \ref{Tab:Algo}.
%
\begin{table}[!htbp]\small
	\caption{Pseudocode of the calibration and forecasting algorithm}
	\label{Tab:Algo}
	\centering
	\begin{threeparttable}
		\begin{tabular}{|l|}
			\hline\noalign{\smallskip}
            \color{green}1. 
             \color{black} Let $N$ be the length of our time series $(V_u)_{u\in[1, N]}$. \\       
           \color{green}2. \color{black} Consider a rolling window of fixed size e.g. $L=12$ values. \\
           \color{green}3. \color{black} Choose the predictive horizon $h\geq 1$. Start from $t=L$. \\			
			\color{green}4. \color{black} \textbf{while} $t\leq N-h$ \\ 
			\color{green}5. \color{black} Take the observations of $V_u$, with $u\in [t-L+1,t]$.\\ 
			\color{green}6. \color{black} Assume firstly that $\mu,\sigma,\alpha$ are constants and calibrate them through a \\
           \quad maximum log-likelihood estimation (MLE) (see the tool \cite{azzalini2015package} and \cite[Section 5.1]{bufalo2022forecasting}). \\
\quad  Denote by $(\widehat{\mu}_t,\widehat{\sigma}_t,\widehat{\alpha}_t)$ such estimates in the present rolling window.

   \\
   \color{green}7. \color{black} Compute the prevision $\widehat{V}_{t+h}:=\E[V_{t+h} \mid V_t]$ through Eq. \eqref{eq:condexp}; \\
			\color{green}7. \color{black} Update $t=t+1$; \\
			\color{green}8. \color{black} \textbf{end} \\
\color{green}9. \color{black} Interpolate the series of calibrated parameters $(\widehat{\mu}_t,\widehat{\sigma}_t,\widehat{\alpha}_t)_{t\in[1, N-h]}$ through 3\\
\quad cubic splines, i.e.
$
f^{\mu}(t)=\sum_{i=0}^3 a^{\mu}_it^i, \quad f^{\sigma}(t)=\sum_{i=0}^3a^{\sigma}_it^i, \quad f^{\alpha}(t)=\sum_{i=0}^3a^{\alpha}_it^i. 
$
\\
\quad and obtain the optimal estimates $(\widehat{a}^{\mu}_i,\widehat{a}^{\sigma}_i, \widehat{a}^{\alpha}_i)_{i\in [0,3]}$. 
\\

\noalign{\smallskip}\hline
		\end{tabular}
	\end{threeparttable}
\end{table}

Note that the assumption of piecewise constant parameters made in calibration procedure in Table \ref{Tab:Algo} row $6.$ is guaranteed by Theorem \ref{th1}.
%
%
Figure \ref{F:par_cal} shows the behavior of the estimated parameters $(\widehat{\mu}_t, \widehat{\sigma}_t, \widehat{\alpha}_t)_{t>0}$, while Table \ref{Tab:est2} reports the fitted parameters $(\widehat{a}^{\mu}_i, \widehat{a}^{\sigma}_i, \widehat{a}^{\alpha}_i)_{i\in [0,3]}$.

\begin{figure}[!htbp]
	\centering
\includegraphics[width=1\textwidth]{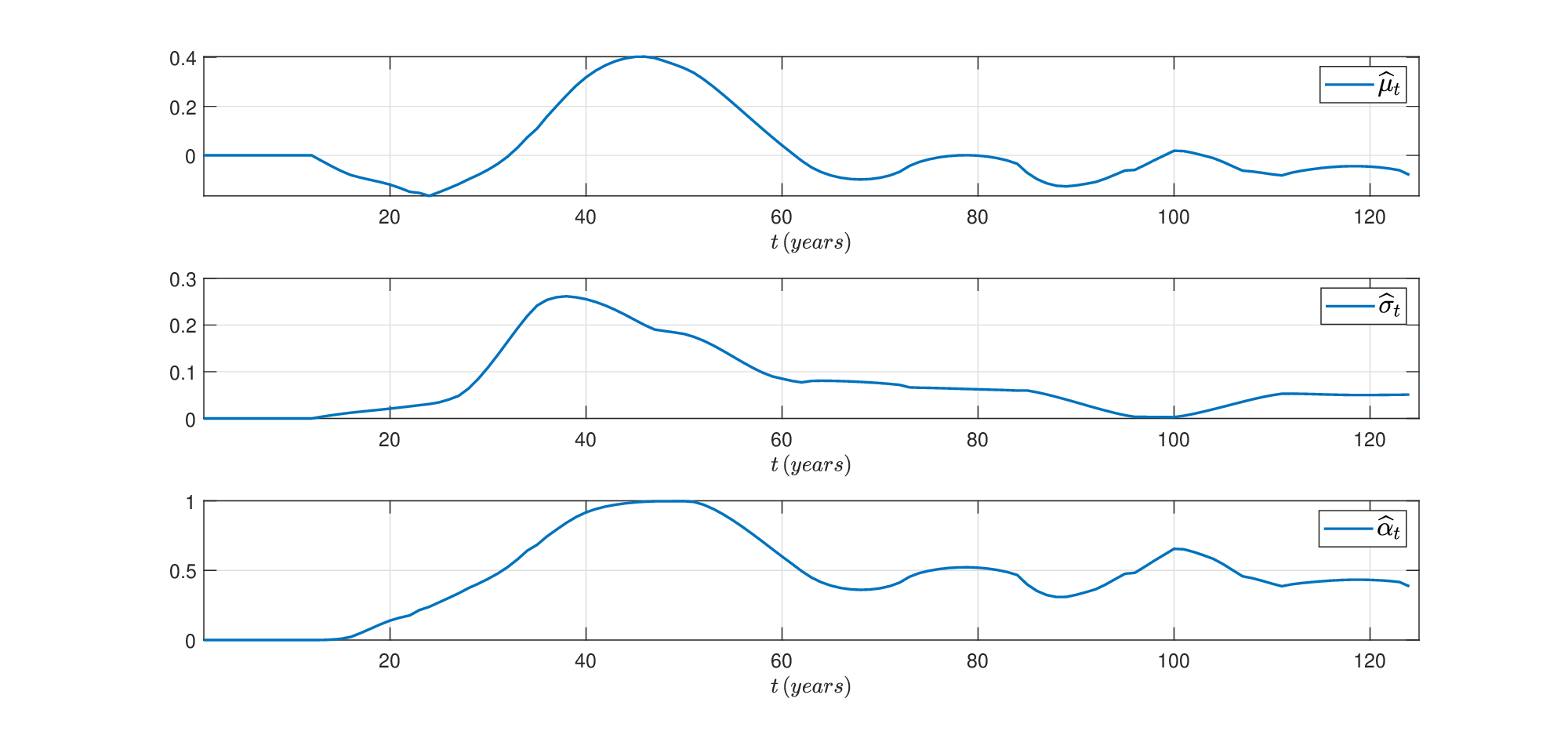}
	\caption{Estimated parameters $(\widehat{\mu}_t,\widehat{\sigma}_t,\widehat{\alpha}_t)_{t>0}$.}
		\label{F:par_cal}
\end{figure}

\begin{table}[!htbp]
\small
\centering
\begin{tabular}{@{}cccc@{}}
\toprule
\multicolumn{4}{c}{Fitted parameters} \\ \midrule
 
\rowcolor{Gray}
\textbf{$\widehat{a}^{\mu}_0$} &
  \textbf{$\widehat{a}^{\mu}_1$} &
  \textbf{$\widehat{a}^{\mu}_2$} &
  \textbf{$\widehat{a}^{\mu}_3$} \\ 
  
\multicolumn{1}{c}{-3.3682$\cdot 10^{-5}$} &
\multicolumn{1}{c}{-6.8061$\cdot 10^{-5}$} &
\multicolumn{1}{c}{-5.4564$\cdot 10^{-4}$} &
\multicolumn{1}{c}{-5.4564$\cdot 10^{-5}$} \\ \midrule
  
\rowcolor{Gray}
\textbf{$\widehat{a}^{\sigma}_0$} &
\textbf{$\widehat{a}^{\sigma}_1$} &
\textbf{$\widehat{a}^{\sigma}_2$} &
\textbf{$\widehat{a}^{\sigma}_3$} \\
  
\multicolumn{1}{c}{0.0160} &
\multicolumn{1}{c}{1.3063$\cdot 10^{-6}$} &
\multicolumn{1}{c}{1.6407$\cdot 10^{-6}$} &
\multicolumn{1}{c}{4.1232$\cdot 10^{-6}$} \\  \midrule

\rowcolor{Gray}
  \textbf{$\widehat{a}^{\alpha}_0$} & 
\textbf{$\widehat{a}^{\alpha}_1$} &
  \textbf{$\widehat{a}^{\alpha}_2$}  &
  \textbf{$\widehat{a}^{\alpha}_3$} \\ 
  
\multicolumn{1}{c}{-6.7491$\cdot 10^{-5}$} &
\multicolumn{1}{c}{-1.0700$\cdot 10^{-4}$} &
\multicolumn{1}{c}{0.0033} &
\multicolumn{1}{c}{0.4632} \\ 
  
  \bottomrule
\end{tabular}
  \caption{Estimated parameters $(\widehat{a}^{\mu}_i,\widehat{a}^{\sigma}_i, \widehat{a}^{\alpha}_i)_{i\in [0,3]}$ through spline interpolation.
  }
\label{Tab:est2}
\end{table}
Figure \ref{Fig:for} displays the relative error between the real data and the forecasted values obtained via Eq. \eqref{eq:condexp} and a deep learning approach (see \cite{qadeer2020long}). 
%
These two procedures yield (normalized) root-mean-squared errors equal to 0.0878 and 0.2237, respectively.

\begin{figure}[!htbp]
	\centering
\includegraphics[width=1\textwidth]{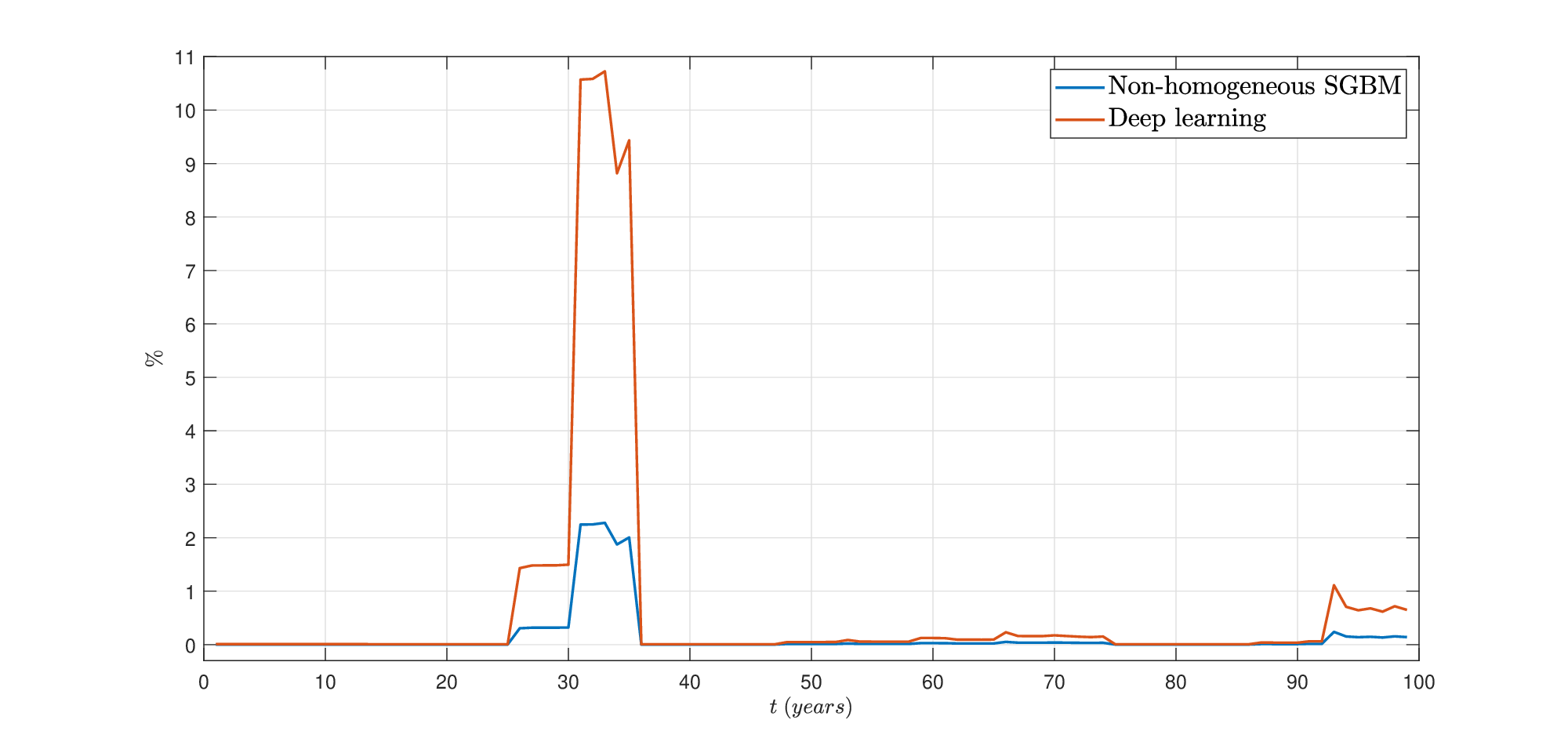}
	\caption{Relative error between real data and forecasted values obtained by Eq. \eqref{eq:condexp} (solid blue line) and a deep learning procedure (red solid line).}
		\label{Fig:for}
\end{figure}

\newpage
\section{Conclusion} \label{Sec:Conclusion}
In this paper, a novel model has been introduced for assessing the impact of natural disasters on affected populations, particularly for loss projection. This model integrates historical data, addresses data skewness, and accommodates temporal dependencies to predict shifts in mortality. To accomplish this, a time-varying skew Brownian motion framework has been presented, with the existence and uniqueness of the solution being established. Within this framework, parameters evolve over time, and past occurrences are integrated via volatility. Pseudocode for calibration has been provided, and a test against a deep learning model has demonstrated the advantages of the proposed approach.

\bibliographystyle{apalike} 
\bibliography{MyBibGBM}

\begin{thebibliography}{}

\bibitem[Azzalini, 1985]{azzalini1985class}
Azzalini, A. (1985).
\newblock A class of distributions which includes the normal ones.
\newblock {\em Scandinavian journal of statistics}, pages 171--178.

\bibitem[Azzalini, 2015]{azzalini2015package}
Azzalini, A. (2015).
\newblock Package ‘sn’.
\newblock {\em The skew-normal and skew-t distributions}, pages 1--3.

\bibitem[Botzen et~al., 2019]{Botzen2019}
Botzen, W. J.~W., Deschenes, O., and Sanders, M. (2019).
\newblock {The Economic Impacts of Natural Disasters: A Review of Models and
  Empirical Studies}.
\newblock {\em Review of Environmental Economics and Policy}.

\bibitem[Bouhadou and Ouknine, 2013]{bouhadou2013time}
Bouhadou, S. and Ouknine, Y. (2013).
\newblock On the time inhomogeneous skew {B}rownian motion.
\newblock {\em Bulletin des Sciences Math{\'e}matiques}, 137(7):835--850.

\bibitem[Bufalo et~al., 2022]{bufalo2022forecasting}
Bufalo, M., Liseo, B., and Orlando, G. (2022).
\newblock Forecasting portfolio returns with skew-geometric brownian motions.
\newblock {\em Applied Stochastic Models in Business and Industry},
  38(4):620--650.

\bibitem[Chen and Cox, 2009]{Chen2009}
Chen, H. and Cox, S.~H. (2009).
\newblock {Modeling Mortality With Jumps: Applications to Mortality
  Securitization}.
\newblock {\em Journal of Risk and Insurance}, 76(3):727--751.

\bibitem[Chen and Cummins, 2010]{Chen2010}
Chen, H. and Cummins, J.~D. (2010).
\newblock {Longevity bond premiums: The extreme value approach and risk cubic
  pricing}.
\newblock {\em Insurance: Mathematics and Economics}, 46(1):150--161.

\bibitem[Corns and Satchell, 2007]{corns2007skew}
Corns, T. and Satchell, S. (2007).
\newblock Skew {B}rownian motion and pricing european options.
\newblock {\em The European Journal of Finance}, 13(6):523--544.

\bibitem[Cox et~al., 2006]{Cox2006}
Cox, S.~H., Lin, Y., and Wang, S. (2006).
\newblock {Multivariate Exponential Tilting and Pricing Implications for
  Mortality Securitization}.
\newblock {\em Journal of Risk and Insurance}, 73(4):719--736.

\bibitem[Cred, 2023]{Cred2023}
Cred (2023).
\newblock {EM-DAT - The international disaster database}.
\newblock [Online; accessed 19. Aug. 2023].

\bibitem[de~Moel et~al., 2015]{deMoel2015}
de~Moel, H., Jongman, B., Kreibich, H., Merz, B., Penning-Rowsell, E., and
  Ward, P.~J. (2015).
\newblock {Flood risk assessments at different spatial scales}.
\newblock {\em Mitigation and Adaptation Strategies for Global Change},
  20(6):865--890.

\bibitem[Deng et~al., 2012]{Deng2012}
Deng, Y., Brockett, P.~L., and MacMinn, R.~D. (2012).
\newblock {Longevity/Mortality Risk Modeling and Securities Pricing}.
\newblock {\em Journal of Risk and Insurance}, 79(3):697--721.

\bibitem[{\ifmmode\acute{E}\else\'{E}\fi}tor{\ifmmode\acute{e}\else\'{e}\fi}
  and Martinez, 2012]{etore2012existence}
{\ifmmode\acute{E}\else\'{E}\fi}tor{\ifmmode\acute{e}\else\'{e}\fi}, P. and
  Martinez, M. (2012).
\newblock {On the existence of a time inhomogeneous skew Brownian motion and
  some related laws}.
\newblock {\em Electronic Journal of Probability}, 17(none):1--27.

\bibitem[Harrison and Shepp, 1981]{harrison1981skew}
Harrison, J.~M. and Shepp, L.~A. (1981).
\newblock On skew {B}rownian motion.
\newblock {\em The Annals of probability}, pages 309--313.

\bibitem[Jonkman et~al., 2008]{Jonkman2008}
Jonkman, S.~N., Bo{\ifmmode\check{c}\else\v{c}\fi}karjova, M., Kok, M., and
  Bernardini, P. (2008).
\newblock {Integrated hydrodynamic and economic modelling of flood damage in
  the Netherlands}.
\newblock {\em Ecological Economics}, 66(1):77--90.

\bibitem[Kahn, 2005]{Kahn2005}
Kahn, M.~E. (2005).
\newblock {The Death Toll from Natural Disasters: The Role of Income,
  Geography, and Institutions}.
\newblock {\em Review of Economics and Statistics}, 87(2):271--284.

\bibitem[Kallenberg and Kallenberg, 1997]{kallenberg1997foundations}
Kallenberg, O. and Kallenberg, O. (1997).
\newblock {\em Foundations of modern probability}, volume~2.
\newblock Springer.

\bibitem[Karatzas and Shreve, 1991]{karatzas1991brownian}
Karatzas, I. and Shreve, S. (1991).
\newblock {\em Brownian motion and stochastic calculus}, volume 113.
\newblock Springer Science \& Business Media.

\bibitem[Kreimer et~al., 2010]{Kreimer2010}
Kreimer, A., Arnold, M., and Carlin (2010).
\newblock {Building safer cities - the future of disaster risk}.
\newblock {\em World Bank}.

\bibitem[Le~Gall, 2006]{le2006one}
Le~Gall, J.-F. (2006).
\newblock One—dimensional stochastic differential equations involving the
  local times of the unknown process.
\newblock In {\em Stochastic Analysis and Applications: Proceedings of the
  International Conference held in Swansea, April 11--15, 1983}, pages 51--82.
  Springer.

\bibitem[Liu and Li, 2015]{Liu2015}
Liu, Y. and Li, J. S.-H. (2015).
\newblock {The age pattern of transitory mortality jumps and its impact on the
  pricing of catastrophic mortality bonds}.
\newblock {\em Insurance: Mathematics and Economics}, 64:135--150.

\bibitem[Qadeer et~al., 2020]{qadeer2020long}
Qadeer, K., Rehman, W.~U., Sheri, A.~M., Park, I., Kim, H.~K., and Jeon, M.
  (2020).
\newblock A long short-term memory (lstm) network for hourly estimation of pm2.
  5 concentration in two cities of south korea.
\newblock {\em Applied Sciences}, 10(11):3984.

\bibitem[Revuz and Yor, 2013]{revuz2013continuous}
Revuz, D. and Yor, M. (2013).
\newblock {\em Continuous martingales and Brownian motion}, volume 293.
\newblock Springer Science \& Business Media.

\bibitem[Ritchie et~al., 2022]{Ritchie2022}
Ritchie, H., Rosado, P., and Roser, M. (2022).
\newblock {Natural Disasters}.
\newblock {\em Our World in Data}.

\bibitem[Weinryb, 1983]{weinryb1983etude}
Weinryb, S. (1983).
\newblock Etude d'une equation diff{\'e}rentielle stochastique avec temps
  local.
\newblock {\em S{\'e}minaire de probabilit{\'e}s de Strasbourg}, 17:72--77.

\end{thebibliography}

\end{document}